\documentclass[letterpaper,USenglish, numberwithinsect, cleveref, autoref, thm-restate]{lipics-v2021}

\hideLIPIcs
\usepackage{amsmath}
\usepackage{amsfonts}
\usepackage{graphicx}
\usepackage{microtype}
\usepackage[dvipsnames]{xcolor}
\usepackage{hyperref}
\usepackage{amsthm}
\usepackage{svg}
\usepackage{comment}
\usepackage{wrapfig}
\usepackage{thm-restate}

\usepackage{cite}
\usepackage{quoting}
\usepackage[section]{placeins}
\quotingsetup{vskip=0pt}

\usepackage[textsize=scriptsize]{todonotes}

\newcommand{\calD}{\mathcal{D}}

\newcommand{\calP}{\mathcal{P}}

\usepackage[normalem]{ulem} 

\definecolor{burntorange}{rgb}{0.75, 0.34, 0} 

\nolinenumbers
\hideLIPIcs
\allowdisplaybreaks

\hypersetup{
  colorlinks = true,
  linkcolor = {MidnightBlue},
  citecolor = {MidnightBlue},
  urlcolor  = {MidnightBlue}
}

\title{Harvesting Brownian Motion: Zero Energy Computational Sampling}

\date{\today}

\author{David {Doty}}{Computer Science, University of California Davis, The United States of America \and \url{https://web.cs.ucdavis.edu/~doty/}}{doty@ucdavis.edu}{https://orcid.org/0000-0002-3922-172X}{NSF grants 2211793, 1900931, and CAREER-1844976.}

\author{Niels {Kornerup}}{Computer Science, University of Texas at Austin, The United States of America \and \url{https://nielskornerup.github.io/}}{nielskornerup@utexas.edu}{https://orcid.org/0000-0002-1519-726X}{}

\author{Austin {Luchsinger}}{Electrical and Computer Engineering, University of Texas at Austin, The United States of America \and \url{https://sites.google.com/utexas.edu/austinluchsinger/home}}{amluchsinger@utexas.edu}{https://orcid.org/0000-0002-8180-9762}{Carroll H.\ Dunn Endowed Graduate Fellowship in Engineering}

\author{Leo {Orshansky}}{Computer Science, University of Texas at Austin, The United States of America}{orshaleo@utexas.edu}{https://orcid.org/0009-0005-4445-7517}{}

\author{David {Soloveichik}}{Electrical and Computer Engineering, University of Texas at Austin, The United States of America \and \url{https://users.ece.utexas.edu/~soloveichik/}}{david.soloveichik@utexas.edu}{https://orcid.org/0000-0002-2585-4120}{Schmidt Sciences Polymath award.}

\author{Damien {Woods}}{Hamilton Institute and Department of Computer Science, Maynooth University, Ireland \and \url{https://dna.hamilton.ie/woods/}}{damien.woods@mu.ie}{https://orcid.org/0000-0002-0638-2690}{European Research Council (ERC) grant 772766 (Active-DNA), and Science Foundation Ireland (SFI) grants 20/FFP-P/8843 and 18/ERCS/5746.}

\authorrunning{D. Doty, N. Kornerup, A. Luchsinger, L. Orshansky, D. Soloveichik, D. Woods}

\ccsdesc[500]{Theory of computation~Random walks and Markov chains}
\ccsdesc[300]{Hardware~Reversible logic}

\keywords{Thermodynamics of computation, random walks, low energy computation, distribution sampling}

\Copyright{David Doty, Niels Kornerup, Austin Luchsinger, Leo Orshansky, David Soloveichik, and Damien Woods}

\acknowledgements{We thank Chris Thachuk and Ho-Lin Chen for discussions about computation driven by an unbiased random walk.}

\begin{document}
\maketitle

\begin{abstract}
The key factor currently limiting the advancement of computational power of electronic computation is no longer the manufacturing density and speed of components, but rather their high energy consumption. 
While it has been widely argued that reversible computation can escape the fundamental Landauer limit of $k_B T\ln(2)$ Joules per irreversible computational step, there is disagreement around whether indefinitely reusable computation can be achieved without energy dissipation. 
Here we focus on the relatively simpler context of sampling problems, which take no input, so avoids modeling the energy costs of the observer perturbing the machine to change its input.
Given an algorithm $A$ for generating samples from a distribution, we desire a device that can perpetually generate samples from that distribution driven entirely by Brownian motion.
We show that such a device can efficiently execute algorithm $A$ in the sense that we  must wait only $O(\text{time}(A)^2)$ between samples. 
We consider two output models:
Las Vegas, which samples from the exact probability distribution every $4$ tries in expectation, and Monte Carlo, in which every try succeeds but the distribution is only approximated.
We base our model on continuous-time random walks over the state space graph of a general computational machine, with a space-bounded Turing machine as one instantiation.
The problem of sampling a computationally complex probability distribution with no energy dissipation informs our understanding of the energy requirements of computation, and may lead to more energy efficient randomized algorithms.
\end{abstract}


\section{Introduction}

\subparagraph*{Motivation}
While time and space are the most commonly used metrics of computational complexity, energy is becoming an increasingly important measure.
The exponential increase in the world's total computational power far exceeds the rate of growth of the energy production capacity\footnote{Recently estimated at 68\% versus 8\% per year, respectively~\cite{frankreversible}.},
and the environmental impact of computation is severe.
For example GPT-3 required $1{,}287$ MWh of electricity to train, resulting in an estimated carbon footprint of around $502$ tons of $CO_2$~\cite{LCL22}.
As of 2022 the information and communications technology sector resulted in around $2\%$ of global $CO_2$ emissions~\cite{LCL22}, and this percentage is likely to continue growing.

Since increasing the efficiency of conventional computation is expected to make a limited impact with such trends, new energy-efficient computational paradigms may become essential~\cite{frankreversible}.
While modern transistor-based computers remain far from the thermodynamic limits of computation, biochemical systems can perform computation where the energy costs are much closer to this regime (e.g. \cite{BAPE+03}) and provide an important proof-of-principle.
Whatever form future low-energy computation takes, the theoretical framework of reversible computing is key for reasoning about the thermodynamics of computation.

\subparagraph*{Reversible computation}
Consider deterministic computation abstracted as a Turing machine (or some more general model). 
Knowing the current configuration fully determines the configuration after the next computational step (successor configuration).
However, the current configuration often does not fully determine the previous configuration (predecessor configuration).
Whenever this happens, a computational step is said to be irreversible.
A typical example of an irreversible operation is erasure in which after a bit is erased we lose the information of whether the bit was 0 or 1.

As argued by Landauer and Bennett, computation has an energy lower bound
of $kT \ln 2$ per irreversible step (or more concretely $kT \ln n$ if there are $n$ possible predecessor configurations,  where the scaling factor $kT$, in units of Joules, is the Boltzmann constant times temperature)~\cite{Lan61,bennett1973logical}.
However, \emph{reversible} computation, in which no step is irreversible, has no such per-step energy dissipation lower-bound.
Moreover, it is possible to reprogram any deterministic computation as a sequence of logically reversible steps, provided that the computation is allowed to save a copy of its input~\cite{Ben03}. 
While general methods of constructing reversible computations from irreversible computation come at some asymptotic cost to the space and/or time complexity of the program, the overhead can be made surprisingly small~\cite{Ben03, LMcT00, SM13, AGS15}.
We include additional justification for considering reversible computation in \cref{sec:discussion}.

Although reversible computation gives a way to overcome the $kT \ln 2$ Landauer limit per computational step, there is some controversy about whether meaningful computation can be performed with zero energy dissipation.
First, Bennett argued that his construction for translating any deterministic Turing machine computation into a (reversible) 1D unbiased random walk ``does not deserve to be called a computation''~\cite{bennett1973logical} since
catching the machine in an output state will occur with negligible probability. 
Bennett gives a simple remedy by  putting a small (constant) bias on each forward step, i.e.~an energetic cost per step (one could instead place a single, deep and energetically expensive, trap at the final state). 
Second, a line of work has argued that {\em reusable} computation without energy dissipation is impossible, even for reversible computing~\cite{norton2013brownian,strasberg2015thermodynamics}.
The basic argument is that, after a computation has reached equilibrium, resetting the computer
to its unique initial state  involves an unavoidable energy cost to decrease entropy. 
This additional cost scales logarithmically with the length of the computation.
Our work will address both of these criticisms. 

While more niche in its applications than function computation, distribution sampling is an important problem. For example sampling large prime numbers is essential for creating cryptographically secure RSA keys and sampling from the co-domain of hash functions is the primary mechanism behind many proof of work schemes \mbox{\cite{RSA78,NS09}}. Distribution sampling can even be a hard problem---in fact boson sampling is the problem behind many recent attempts at quantum supremacy \cite{AA10}.

While we are not aware of zero-energy sampling of arbitrary probability distributions that are defined with respect to a sampling \emph{algorithm}, it is known that arbitrary distributions can be generated without energy dissipation~\cite{owen2019number,cappelletti2020stochastic}.
However, the distribution is ``hard-coded'' and no computation is performed; such constructions do not give guarantees about the time-complexity of sampling, nor device complexity, as a function of the generating algorithm.
In contrast, our constructions involve stochastic processes that execute computation, enabling us to upper bound the time complexity required to produce our samples in terms of the computation's time-complexity.

\subparagraph*{Our results}
We design machines that only change state according to the random perturbations of Brownian motion.
By itself such a machine may not seem very useful---even if the unbiased random walk causes the machine to perform a useful computation---intuition suggests it will not be able to ``lock in'' its answer.
Nonetheless, we design machines that are reversible and adiabatic 
in a way where after giving the random walk sufficient time to explore the computation graph, an observer can measure its state and will likely find that some useful computation has been completed (specifically, with constant probability per computation). 
Whereas  Bennett~\cite{bennett1973logical} added an energetic bias to drive his reversible computation to the output state, we avoid any such energetic cost with a (mere) factor 2 increase in state space/computation time. 
Specifically, we design machines where either: (1) after sufficient time the machine will have a constant probability of producing output that is distributed i.i.d. according to our desired probability distribution, or (2) after sufficient time the machine will always be in a state containing samples from a distribution close to the target.
We call these kinds of machines Las Vegas and Monte Carlo respectively in homage to the similar categories of randomized algorithms first defined in~\cite{Els80}.
The time to wait between samples scales as $O(\text{time}(A)^2)$, where $A$ is a reversible Turing machine generating samples from the desired probability distribution given uniform-random input strings.

Similar to the case with Maxwell's Demon, observing our zero energy machine and any post-processing of the observed information will require energy, a cost that seems unavoidable if one wishes to read and record samples.  
However, both our Las Vegas and Monte Carlo machines require small energy expenditure on behalf of the observer.
In the Las Vegas construction, the observer will measure the machine at specified intervals until (only) a single trit indicates that an output value is present, which will occur after each measurement with constant probability.
In the Monte Carlo construction, the observer's measurements will always contain an output value that is close to the desired distribution (i.e.\ zero additional memory storage required for the observer).

\subsection{Related work}

\subparagraph*{No Go Theorems for Brownian Distribution Sampling}\label{sec:no-go}
Prior works (e.g., \cite{norton2013brownian,strasberg2015thermodynamics}) have provided no go theorems suggesting that computation has an unavoidable thermodynamic cost.
While these results are obstacles that rule out some strategies for Brownian computation, they do not apply to our constructions.

In \cite{norton2013brownian} the authors object that (1) without energy bias, a (reversible) computation that enters $t$ states will only have probability density $1/t$ on the terminal state at thermodynamic equilibrium, and (2) the free energy released by allowing the machine to hit thermodynamic equilibrium from its initial configuration (at temperature $T$) is at least $kT \ln t$.
Problem (1) can be addressed by adding a free energy trap to the terminal state; however, the authors note that this further increases the free energy released when the machine hits thermodynamic equilibrium.
Problem (2) is potentially much more problematic, as any machine that needs to be reset before it can be reused must expend $kT \ln t$ energy to do so.
This problem is potentially much worse in the setting of distribution sampling, as there is additional free energy released by the process due to the entropy inherent in the distribution being sampled.

{We overcome both problems: problem (1) by designing a machine that exploits redundancy in its computation and problem (2) by the machine not needing to be reset. 
In other words, our construction does not require any external manipulation of its state, in particular of neither input nor output states, between recorded samples.}
The key insight from our work is a mechanism for obtaining multiple samples from a desired distribution without needing to perturb the system from thermodynamic equilibrium.
In fact, our constructions behave better the closer the system gets to reaching thermodynamic equilibrium between our samples.
The only thermodynamic costs associated with our constructions are those required of an observer to record a small (constant per sample) number of bits of metadata about the state of our machine and the sample from the target distribution.

While \cite{norton2013brownian} argue that energy expenditure logarithmic in the computation time is unavoidable, the authors of \cite{strasberg2015thermodynamics}  make an even stronger claim that constant energy must be dissipated per step---assuming the duration of the computation is not known ahead of time.
In \cite{strasberg2015thermodynamics} the authors argue that this is a natural assumption to make, as it is in general impossible to know the number of steps before a machine halts.
We note that this objection is not applicable to most reasonable machines that perform distribution sampling, which will often have asymptotically proven upper bounds on their running times.
Even if the number of steps for a machine is unknown, our construction's version of a free energy trap can be easily modified to scale with the time of the computation as long as the machine takes the same number of steps for all choices of its randomness.\footnote{In fact if we do not care about the time between samples for our construction, we can force our construction to simulate all choices of randomness to artificially meet this requirement. See \cref{sec:discussion} for details.}

\subparagraph*{Molecular computing}
Distribution sampling by chemical means has received recent theoretical attention.
Work showing systematic ways to build up probability distribution includes ``one-shot'' arbitrary probability generators that---unlike our constructions---cannot be reused~\cite{fett2007synthesizing},
as well as constructions generating the desired distribution at the equilibrium of the chemical reaction network~\cite{cardelli2018programming}. 
Although computing at the equilibrium of the chemical reactions, these reaction equations don't satisfy detailed balance and thus cannot be implemented without energy dissipation. 
More recently it was shown that detailed balanced systems of chemical reactions can produce arbitrary probability distributions without a power supply~\cite{cappelletti2020stochastic}.
However, this work does not provide a method to embed arbitrary probability generating computation and does not connect the time complexity of this computation to the interval of sampling.

\subparagraph{Partially irreversible algorithms}
Our paper seeks to understand the degree to which energy use can be driven all the way to zero in a reversible algorithm,
showing how to dispense with the small (but positive) energy use per step that Bennett~\cite{bennett1973logical} used to increase the probability of seeing an output state.
Another line of work on energy-efficient computation in some sense goes in the opposite direction:
allowing some logically irreversible steps, conceding that each one costs some energy, while seeking to minimize the number of irreversible steps~\cite{demaine2016energy, tyagi2016toward, demaine2021efficient}.

\subparagraph*{Road map}
In \cref{sec:prelim} we discuss our formal models for general machines, kinetics, and thermodynamics of computation.
In \cref{sec-LV} and \cref{sec-MC} we present our constructions applied at a high-level to a general computational device.
We make these constructions more explicit for the language of Turing machines in \cref{sec:concrete}.
We conclude the paper with discussion and open questions in \cref{sec:discussion}.

\section{Preliminaries}\label{sec:prelim}
\subsection{General machines}

A general machine is an abstract model of computation that lets us discuss the relationship between the syntax of the machine's configurations and the semantics of its operation.
Because we are interested in a physically realistic, bounded device, we define computation with respect to a finite configuration space.

\begin{definition}
A \emph{general machine} $M$ is a tuple $(\Sigma, f)$, 
where $\Sigma$ is a finite set of possible configurations, $f: \Sigma \to \calP(\Sigma)$ is the transition function, and $\calP(\Sigma)$ is the power set of~$\Sigma$.
Each configuration in $\Sigma$ is a tuple of the form $(w,m,o)$, where $w$ is the internal work register, $m$ is the metadata register, and $o$ is the output register.
\end{definition}

Although any particular general machine has a finite configuration space,
to perform computation on arbitrary inputs
we can define a family $M_n = (\Sigma_n, f_n)_{n \in \mathbb{N}}$ of general machines for inputs of increasing size $n$.
Our subdivision of configurations $\Sigma$ into tuples of the form $(w,m,o)$ (aka registers) gives some structure that will be useful when designing general machines.
The above definition is meant to be general enough to encompass any reasonable notion of computation on a finite, discrete state space.
For example, a space-bounded Turing machine is a general machine where $\Sigma$ is the set of all configurations of the (finite) tape and head, and $f$ encodes the transition rules of the machine.
We could define the work, metadata, and output registers of the configurations in many different natural ways; for example metadata could be encoded in the head state and the output register could be a particular region on the Turing machine tape (or an output tape for multi-tape machines).

Note that our general definition allows arbitrarily large changes to the configuration to occur in one step, which may be unrealistic. 
It thus will be important that there is a ``local'' model, such as a Turing machine that underlies the general machine. Indeed, we take care to provide such a model in \Cref{sec:concrete}.

Intuitively a general machine is deterministic if given any initial configuration $I$, for all $t\in\mathbb{N}$, there is a unique configuration that the machine will be in after $t$ time-steps.
As in~\cite{Ben89}, we define computation as reversible if each configuration has a unique predecessor and successor. 
Formally:
\begin{definition}
We say a general machine $M$ is \emph{deterministic} if every set in the codomain of the transition function $f$ has at most one element---that is each configuration has at most one successor configuration.
We say that a general machine $M$ is \emph{reversible} if it is deterministic and $f$ is injective---that is each configuration also has at most one predecessor configuration.
\end{definition}
Equivalently, a machine is reversible if there exists a deterministic reverse machine $(\Sigma,f^{-1})$, where $f^{-1}(f(C))=f(f^{-1}(C))=C$ for all $C\in\Sigma$.

\begin{definition}
Configuration $C \in \Sigma$ is a \emph{terminal configuration} if
$f(C) = \emptyset$. That is $C$ has no successor configurations.
\end{definition}

We define the (discrete time) dynamics of a general machine $M$ as follows.
The machine starts at time $t=0$ in some initial configuration $I \in \Sigma$. 
During each time-step, if $M$ is in some non-terminal configuration $C$, in the next time-step it will adopt a uniformly random configuration in the set $f(C)$.
If instead $C = (w, m, o)$ is terminal, then the machine remains in this configuration, and we say that the machine halted with output $o$.\

\subsection{Energy and computation}
Thus far, we have discussed computational machines abstractly in terms of their logical state space. 
For physical realizations, we make the following somewhat standard assumptions~\cite{ouldridge2018importance}.
Any physical realization of a general machine will necessarily map each logical configuration to a physical configuration, each having a (potentially different) amount of free energy $\Delta G$.
We also assume that the machine is submerged in a sufficiently large external heat reservoir such that its internal temperature remains constant, keeping the values of $\Delta G$ constant even after the machine undergoes transitions that consume or release heat.

Further we make the assumption that the physical state of the realization of a general machine will be completely determined by the logical state. 
In other words, the physical realization will have the same state space as the general machine (with any additional degrees of freedom captured in the free energies of configurations).
In \cref{append:thermodynamics}, we formally connect our abstract computational model to thermodynamic and kinetic energy models.
Here we give a more direct definition for an \emph{adiabatic} (constant energy) machine and describe how such machines evolve over time.
These statements are sufficient to describe the behavior of our constructions.

\begin{definition}
    Let $M = (\Sigma, f)$ be a general machine.
    The adiabatic machine for $M$ is an undirected graph $P_M = (V, E)$ with vertices $V = \Sigma$ and edges $E = \{(U,W) \mid W \in f(U) \}$ (i.e., if $W$ is a successor configuration of $U$ then they are connected by an edge). 
\end{definition}

An adiabatic machine is formally defined as a graph where the states are represented as nodes and transitions as edges.
As described in \cref{append:thermodynamics}, we define the evolution of such a machine using the language of continuous time random walks on graphs.

\begin{restatable*}{corollary}{adiabaticevolution}
    Let $P_M$ be an adiabatic machine. $P_M$ evolves via an unbiased random walk and converges to a uniform stationary distribution.
\end{restatable*}

\subsection{Distribution Sampling}
\label{sec:distribution-sampling}

In this paper we build machines that sample from a desired distribution.
This is a more straight-forward problem than function computation, as we do not need to worry as much about energy dissipation associated with changing the input.

A standard way to define a distribution sampler (or random number generator) is as a probabilistic algorithm which has access to fair coin flips and output samples from distribution $\mathcal{D}$ when executed. 
It is well-known that the randomized algorithm's randomness can be moved ``up-front'' as input to a deterministic algorithm. Thus we define:

\begin{definition}
    Let $\mathcal{D}$ be a probability distribution over finite domain $Y$.
    We say a function $g: \{0,1\}^r \to Y$ is a \emph{$\mathcal{D}$-generator} if $g(x) \sim D$ when $x$ is chosen uniformly at random. 
\end{definition}

Note that for all distributions other than uniform, a $\mathcal{D}$-generator cannot be directly computed in a reversible manner since that would require $g$ to be one-to-one.
We can get around this restriction by instead reversibly computing a function $h(x) = (x,g(x))$, which is reversible because it preserves the input $x$.

In theory one could---given a machine $M$ that computes a $\calD$-generator function $h$---run $M$ and pre-compute
the value of $h$ on all inputs in $\{0,1\}^r$.
Doing so would allow the construction of an 
energy and time efficient device that samples from $\calD$ via a lookup table,
albeit physically of exponential (in $r$) size.
This is similar to how all boolean functions have an exponential sized constant depth circuit.
As such we will only be concerned with the task of designing devices that are constructed from a description of $M$ in time that is polynomial in both $r$ and the size of $M$.
%
\section{Las Vegas Sampling}\label{sec-LV}
We first give our formal definition of what it means for a general machine to be a Las Vegas sampler.
Intuitively, we can generate the exact probability distribution at the cost of having to reject some samples as invalid.

\begin{definition}
Let $M = (\Sigma, f)$ be a general machine where each configuration $(w,m,o) \in \Sigma$ has metadata $m \in \{-1, 0, 1\}$.
We say that a physical embedding of $M$ is a \emph{Las Vegas sampler} for a distribution $\calD$ if the accepted outputs of the following procedure starting from any initial configuration are independent samples from $\calD$ for any $T > 0$:
\begin{quoting}[indentfirst=false]
Measure the metadata $m_\mathrm{prev}$ of the current configuration every $T$ time units until $m_\mathrm{prev} \neq 0$.
Then measure the metadata $m$ of the current configuration every $T$ time units  and do the following.
If $m = -m_\mathrm{prev}$,
then let $m_\mathrm{prev} = -m_\mathrm{prev}$ and accept the output of the current configuration;
otherwise, reject.
\end{quoting}
\end{definition}

\begin{definition}
A Las Vegas sampler is \emph{$(T,\delta)$-efficient} if starting from any configuration and waiting $T$ time units, the probability of measuring metadata 
 $m = -1$ is at least $1-\delta$ and the probability of $m=1$ is at least $1-\delta$. 
\end{definition}



Note that if a Las Vegas sampler is $(T,\delta)$-efficient, then:
(1) the initialization phase before we find $m_\text{prev} \neq 0$ takes longer than $c T$ time with probability at most $\delta^c$;
(2) the time between successive accepted outputs is longer than $c T$ time with probability at most $\delta^c$.


\subsection{The general machine}

In this section we describe how to construct general machines with physical embeddings that are an efficient Las Vegas sampler for desired distributions. 
Importantly, these general machines have the same configuration space and connectivity as the more detailed Turing machine-level construction we develop in \Cref{sec:concrete}, and thus the Turing machine is an instantiation of our general machine. 
Introducing the more abstract general machine first significantly eases understanding.
Our construction is divided into three parts: the randomizer, the computation, and the output holding regions;
below we describe at a high level how these parts are designed for a general machine whose physical embedding is a Las Vegas sampler.

As a starting point we take the configuration space of a reversible Turing machine $M$ computing a $\mathcal{D}$-generator as described in \Cref{sec:distribution-sampling}.
The requirement that the Turing machine be reversible is not restricting as any such machine can be converted to a reversible one with only polynomially more time steps~\cite{Ben89}.
As the Turing machine is reversible, its configuration space is a collection of chains.
We assume that $M$ runs in at most $T$ steps, by which we mean that the longest length of a chain in $M$'s configuration space is at most $T$.

Our eventual complete general machine construction is a layered graph in the sense that the nodes are partitioned such that nodes in layer $L_i$ only have edges to layers $L_{i-1}$ and $L_{i+1}$.
A simple example of our construction for the trivial machine $M$ with a single bit of input can be seen in \cref{fig:LV-simple}, and the full general machine construction is illustrated schematically in \cref{fig:LV-general}.
Below we describe how these general machines are constructed, assuming we start with a reversible Turing machine $M$.

\subparagraph*{Adding the output holding region}
To have a constant probability of producing a new output after measurement, it is essential that our computation's graph has a constant fraction of its configurations containing the output.
We achieve this by artificially boosting the length of $M$'s computation by adding a tail of redundant configurations containing the output.
Specifically, we elongate `short' computations: we  modify $M$ by elongating paths so that all the chains in its state space have length exactly $T$.
We then further extend the length of each chain $c$ with $2T+r+1$ nodes that all contain the value in the output register from the terminal configurations of $c$.
We let $M'$ denote the resulting general machine.

\subparagraph*{The randomizer}
Let $r$ be the number of random bits that machine $M$ starts with on its random bitstring input tape. We next design a {\em randomizer} that allows $M'$ to  change its random input in order to sample from the correct distribution.
Specifically the randomizer's configuration graph has  two outer {\em sides} and has the property that any random walk over its state space that starts on one {\em side} and arrives at the {\em other side} will end in a state that has a uniformly random bitstring in a work register.
We design the randomizer to have this property by letting it randomize each bit of any initial $r$ bit configuration during a different step.
We achieve this by using the $r+1$-dimensional butterfly graph as the configuraiton graph and assigning the source (left side) and sink (right side) nodes of this graph to length $r+1$ bitstrings.

The randomizer, shown in \cref{fig:LV-simple}, is composed of $r+1$ columns that each contain $2^{r+1}$ nodes.
Each column of nodes contains one node whose work register contains each length $r+1$ bit-string.
There are edges between the nodes in columns $i$ and $i+1$ where the bitstring in their work registers are either the same or differ only on the value of the $i$'th bit.

\begin{lemma}\label{lem:randomizer-independent-samples}
    An unbiased random walk that starts on the left of this graph and ends on the right will end up in a node whose work register contains $r+1$ bits of which the first $r$ are uniformly random.
\end{lemma}
\begin{proof}
    Let $p$ be any path through this graph that starts on the left and ends on the right.
    Starting from column~0, let $t_1, \ldots t_r$ be the last time-instant where the random walk enters each column $1,\ldots,r$ . We note that the choice of edge taken at time $t_i$ determines the value of bit $i-1$ assigned to the node reached at the end of the walk.
    Since this is an unbiased random walk, the choice of which edge is taken during each of these time-steps is uniformly random.
    Thus the random walk will end with a uniformly random bitstring of length $r$, plus one ancillary bit.
\end{proof}

\subparagraph*{Linking the randomizer and the computation}
Observe that there is one node on either side of the randomizer corresponding to each bit-string of length $r+1$, which gives us two nodes corresponding to each bit-string $b$ of length $r$ by ignoring the last bit.
As shown in \cref{fig:LV-general}, on each side of the randomizer, for each bitstring $b$ of length $r$, we add two copies of the chain in $M'$ corresponding to random input $b$ and connect their starting configurations to both of the nodes in the randomizer corresponding to value $b$.
We finally add edges that link the node for step $t$ in each of these chains to the node for step $t+1$ in the other chain.
Doing this gives us a layered graph where each node has in-degree and out-degree two in each direction, so a random walk on this graph behaves like a one-dimensional random walk on the layers.
We assign metadata value $m=-1$ to all the nodes we added in this process to the left of the randomizer that appear after chain link $T$ (where we know the nodes contain samples from $\calD$ in their output values) and metadata value $m=1$ to the same nodes for the chains on the right of the randomizer.
All other nodes have value $0$.

\begin{figure}
    \centering
    \includegraphics[width=.5\textwidth]{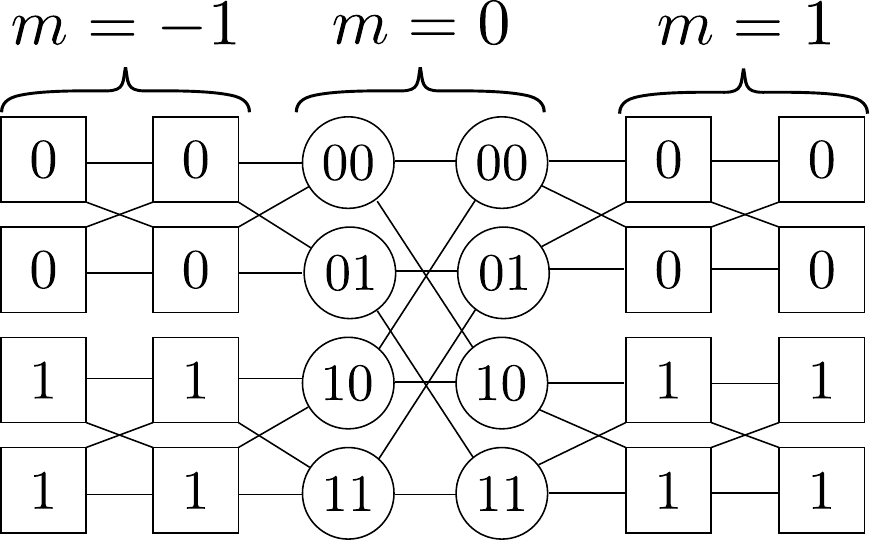}
    \caption{Las Vegas construction, $M'$, applied to a simple Turing machine $M$ with no transition rules and a single bit of input (i.e.~$r=1$). 
    Circles: randomizer with one random bit ($r=1$, for input), one ancillary bit and with metadata $m=0$. 
    Squares: output holding regions labeled with their output bit with metadata $m=-1$ or $m=1$.}
    \label{fig:LV-simple}
\end{figure}

\begin{figure}[b]
    \centering
    \includegraphics[width=\textwidth]{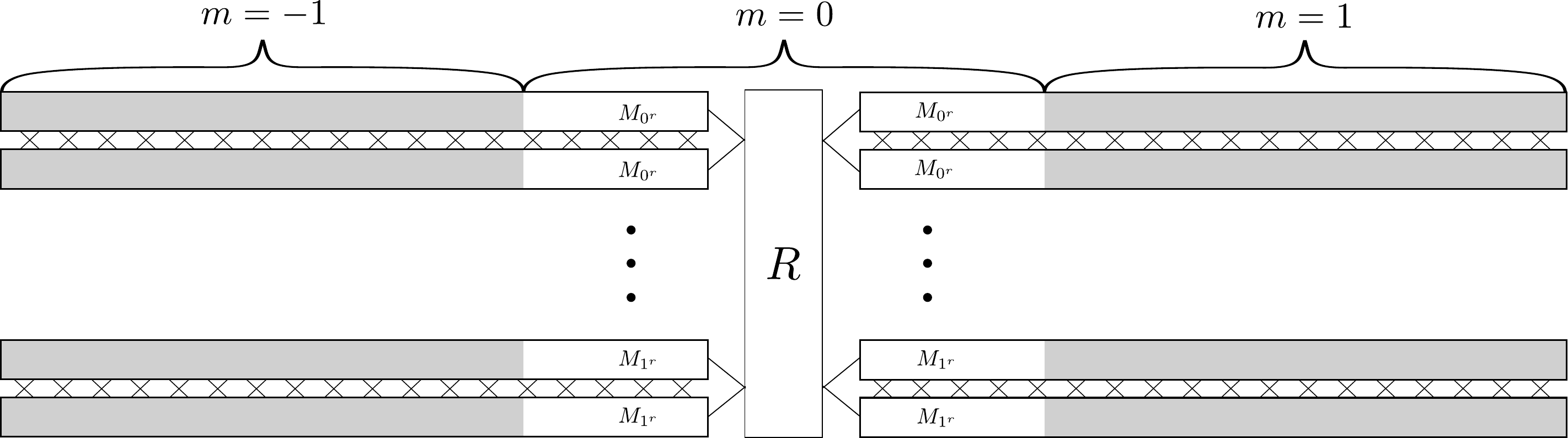}
    \caption{Our general Las Vegas construction $M^\ast$ from a general machine $M$. $R$ is the randomizer which generates random binary input words $b$,  and $M_b$ represents the chain of states for $M$ on input $b$. The grey area represents the output holding states.}
    \label{fig:LV-general}
\end{figure}

\subsection*{The sampling procedure}
We let $M^*$ be the machine constructed above.
We will now give a general strategy for an observer to measure~$M^*$ to get independent samples from the distribution~$\calD$.
The first sample is obtained by waiting time  $\Theta(T^2+r^2)$ between measurements of the metadata register, iterating until a non-zero metadata value is observed.
Then the value of the output register is recorded as the first sample and the observer remembers the last observed metadata bit when it took a sample.
The observer again waits time $\Theta(T^2 + r^2)$ between measurements until the metadata register is $-1$ times the value it had during the last recorded sample, and then records the value of the output register.

The state space of our construction is a graph with diameter $6T + 3r + 3$. Thus waiting time $\Theta(T^2 + r^2)$ between samples is sufficient for there to be a $\geq 1/4$ probability of observing the desired metadata on measurement.
\begin{theorem}\label{lem:gen-lv-sampler}
    Let $M$ be a general machine that samples from a distribution $\calD$. Then applying the above construction and procedure to $M$ yields a general machine $M^*$ whose physical embedding is an adiabatic $(\Theta(T^2+r^2),3/4)$-efficient Las Vegas sampler for $\calD$. 
\end{theorem}
\begin{proof}
    We note that the machine from the above construction yields a state space where the nodes can be partitioned into layers such that the nodes in layer $i$ are only neighbors of the nodes in layers $i-1$ and $i+1$.
    Additionally each node in an internal layer has two edges to nodes in the layer before and after it.
    Thus, the machine will evolve according to an unbiased random walk over the layers of the graph, which is an unbiased walk on a line of length $6T + 3r + 3$.
    Starting from any node on this line, a continuous time random walk will have probability density at least $1/4$ on the layers where nodes have the desired metadata value after time that is $\Theta(T^2 + r^2)$.
    Thus by waiting this amount of time between measurements, we get a $\geq 1/4$ probability that each measurement will give the desired metadata value.
    After $c$ measurements, the probability that one gave the correct metadata is at least $1-(3/4)^c$.
\end{proof}
\cref{fig:plots} shows that in simulation, a general machine built with the above procedure would actually behave like a Las Vegas sampler. 
Our software tool\footnote{\url{https://github.com/leoorshansky/reversible_sim}} can simulate the physical embedding of our Las Vegas and Monte Carlo generation machine constructions based on an arbitrary input reversible Turing machine specification.
\section{Monte Carlo Construction}\label{sec-MC}
Recall that a Las Vegas sampler lets us sample the exact probability distribution we desire, at the cost of having to reject some outputs if they are measured at the ``wrong time.''
In contrast, a Monte Carlo sampler lets us produce a sample from every measurement, 
but we only approximate independent samples from the desired distribution $\calD$.

\begin{definition}
    Let $M=(\Sigma, f)$ be a general machine where each configuration $(w,m,o) \in \Sigma$ has metadata $m \in \{0,1\}$. We say that the physical embedding of $M$ is a \emph{$(T,\epsilon)$-efficient Monte Carlo sampler} for a distribution $\calD$ if for any $c > 0$ the next output of the following procedure, conditioned on any choice of previous outputs, comes from a distribution $\calD'$ such that the total variation distance between $\calD$ and $\calD'$ is at most $\epsilon^{c}$:
    \begin{quoting}[indentfirst=false]
        Wait $cT$ time units and measure the metadata $m$ of the current configuration.
        Then if $m = 0$ record the first value stored in the output register.
        If $m=1$ instead record the second value stored in the output register.
    \end{quoting}
\end{definition}

Unlike in the Las Vegas case, a Monte Carlo sampler will always have two values in its output register---one of which will contain a sample from $\calD'$.
Reading the metadata lets the observer pick which part of the output register it should read before producing its sample.
Unlike with the Las Vegas sampler, there is no guarantee that samples generated by the Monte Carlo sampler will be independent of one another.
Thus a Monte Carlo sampler will exchange sampling exactly from distribution $\calD$ for the guarantee that each measurement will produce a sample.

\subsection{The general machine}
In this section we describe a general machine to make an efficient Monte Carlo sampler.
In \cref{sec:concrete} we give a more explicit version of this high-level construction assuming that the starting machine is a Turing machine.

Similar to the Las Vegas construction, we will assume that we are given a reversible machine $M$ that starts in states where an $r$ bit string serves as an input for determining the sample from distribution $\calD$.
As stated before, the state-space of $M$ is composed of $2^r$ chains representing the computation on each input.
We will think of these chains as having the input on the left and the terminal state on the right.
At a high level, our machine $M^*$ will be composed of two sub-machines $M^*_1$ and $M^*_{-1}$ so that whenever $M^*_1$ takes a step forwards, $M^*_{-1}$ takes a step backwards.
Thus by carefully designing $M^*_{1}$ and $M^*_{-1}$ we can guarantee that one of the sub-machines always has a sample from $\calD'$.

\subparagraph*{The output holding region}
\begin{wrapfigure}{R}{0.3\textwidth}
    \centering
    \includegraphics[width=.3\textwidth]{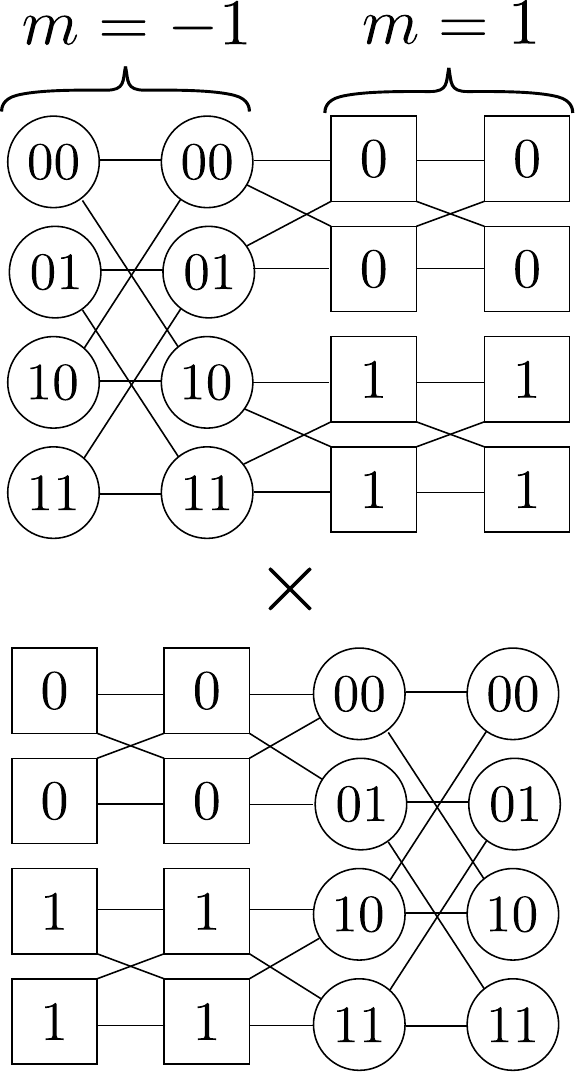}
    \caption{Our Monte Carlo construction applied to a general machine with no transition rules and a single bit of input. Each node in the graph is a pair of nodes in the same column of the sub-machines.}
    \label{fig:MC-simple}
    \vspace{-40px}
\end{wrapfigure}
Similar to our Las Vegas construction, we will modify $M$ by artificially elongating the length of shorter chains so that they all have the same length $T$ by adding nodes with the same value in the output register to the end of each chain.
Then we add $T+r+1$ additional nodes to each chain that maintain the output value and call these nodes the output holding region.
After these changes, we call the machine $M'$.

\subparagraph*{The randomizer}
We use the same randomizer as in the Las Vegas construction to enable the machine to change its inputs.
However, unlike in the Las Vegas construction, we only link two copies of $M'$ to the right hand side of the randomizer.
We again add the edges between layers for the two identical copies of $M'$ so that each node in the graph that is not on the leftmost or rightmost layer has two edges to the left and two edges to the right.
We assign meta-data value $1$ to the nodes in the output holding regions of the chains and $0$ everywhere else.
This gives us a $2T+2r+2$-layered machine where a random walk over the machine behaves like an unbiased random walk over the layers.
The machine ``resets'' its randomness whenever the random walk reaches the leftmost layer, enabling the next measurement to be independent of the previous one.
We will call this machine~$M^*_1$.

\subparagraph*{The sub-machines}
Machine $M^*_1$ is a $2T+2r+2$-layered general machine where the nodes with meta-data value zero contain a sample that is ready to read.
We will define machine $M^*_{-1}$ to have the same state-space as $M^*_1$ except that the output holding regions have meta-data value $-1$ instead of $1$.
Our machine $M^*$ will be created by joining nodes of $M^*_1$ and $M^*_{-1}$.
Specifically for each node $v^1_i$ in layer $i$ of $M^*_0$ with configuration $(w_1,m_1,o_1)$ and each node $v^{-1}_j$ in layer $T+r-i$ with configuration $(w_{-1}, m_{-1}, o_{-1})$ of $M^*_1$
we create a single node $v_{i,j}$ in $M^*$ whose registers are in the configuration $((w_1,w_{-1}), (m_{1} + m_{-1}), (o_1,o_{-1}))$.
We then add edges between nodes $v_{i,j}$ and $v_{i',j'}$ exactly when $v^{1}_i$ has an edge to $v^1_{i'}$ and $v^{-1}_j$ has an edge to $v^{-1}_{j'}$.
Note that the nodes of $M^*$ always have meta-data value $-1$ or $1$ and so there is only a single bit of meta-data.
That bit directly corresponds to which machine is in an output holding configuration.
An example of this construction is in \cref{fig:MC-simple} and the general version in \cref{fig:MC-general}.

\subsection*{The sampling procedure}
\begin{theorem}\label{lem:gen-mc-sampler}
    If the original machine had chain length at most $T$, then the above procedure gives an adiabatic $(\Theta(T^2+r^2),1/2)$-efficient Monte Carlo sampler.
\end{theorem}
\begin{proof}
    After time $\Theta(T^2 + r^2)$ there is an $1/2$ probability that a random walk on a line of length $2T+2r+2$ has hit both the leftmost and rightmost nodes.
    Thus time $\Theta(T^2 + r^2)$ is sufficient so that there is a $1/2$ probability that a random walk on the state space of our Monte Carlo sampler has hit a node on the leftmost and rightmost layers.
    For any $c>0$ after waiting time $\Theta(cT^2 + cr^2)$ the probability that this has happened is at least $1-1/2^c$.
    After any random walk over the state space has done both of these, the sample observed by measuring the machine will come from the distribution $\calD$ as it will be independent of what was observed during prior measurements.
    Therefore the total variation distance between $\calD$ and the samples given by this machine when waiting time $\Theta(cT^2 + cr^2)$ is at most $1/2^c$ and this is an $(\Theta(T^2 + r^2), 1/2)$-efficient Monte Carlo sampler.
\end{proof}
\cref{fig:plots} shows that in simulation, this construction behaves as expected. 
\begin{figure}
    \centering
    \includegraphics[width=.45\textwidth]{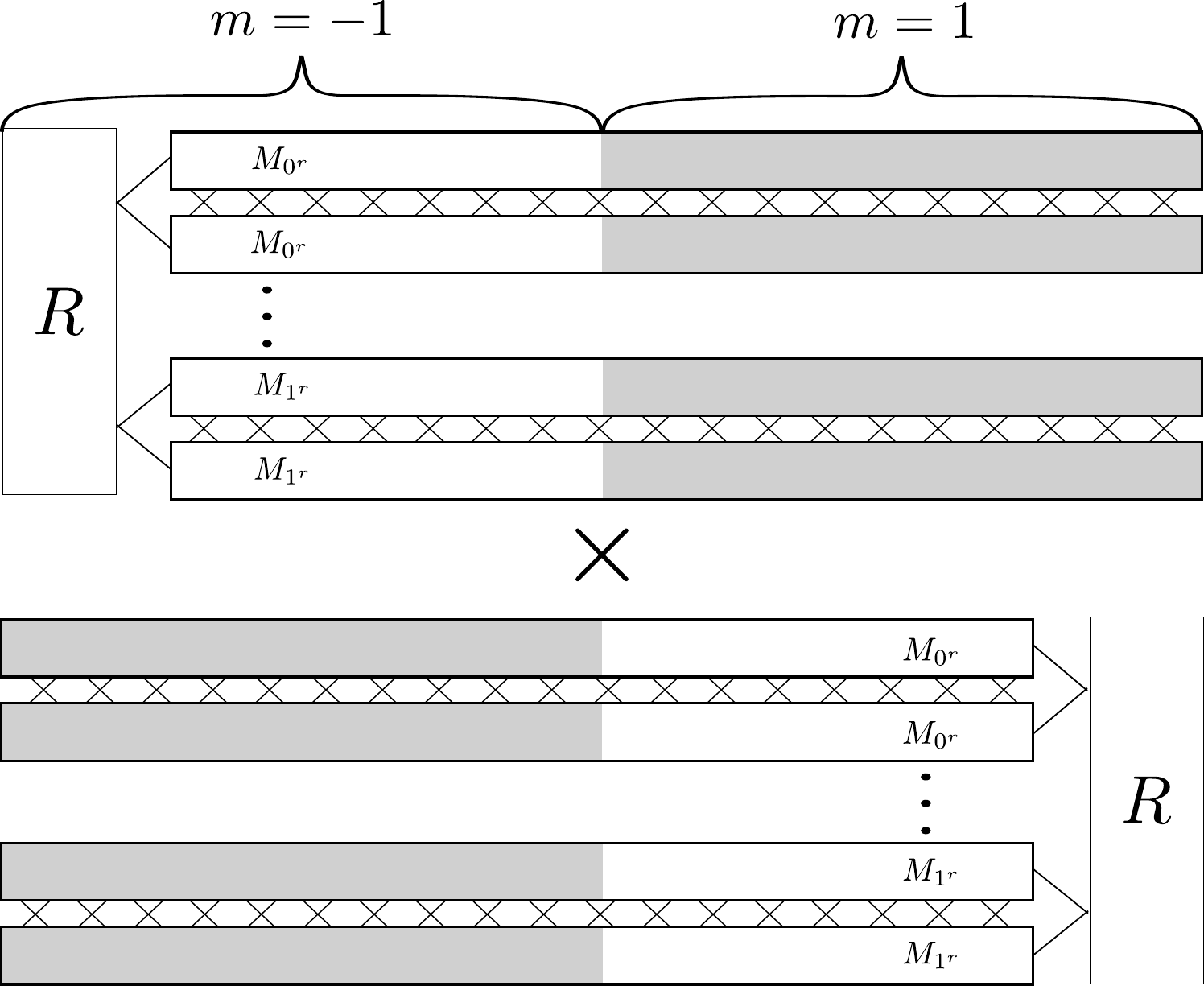}
    \caption{Monte Carlo: Each node in the graph is a pair of nodes from the sub-machines.}
    \label{fig:MC-general}
\end{figure}
\begin{figure}[b]
    \centering
    \includegraphics[width=.85\textwidth]{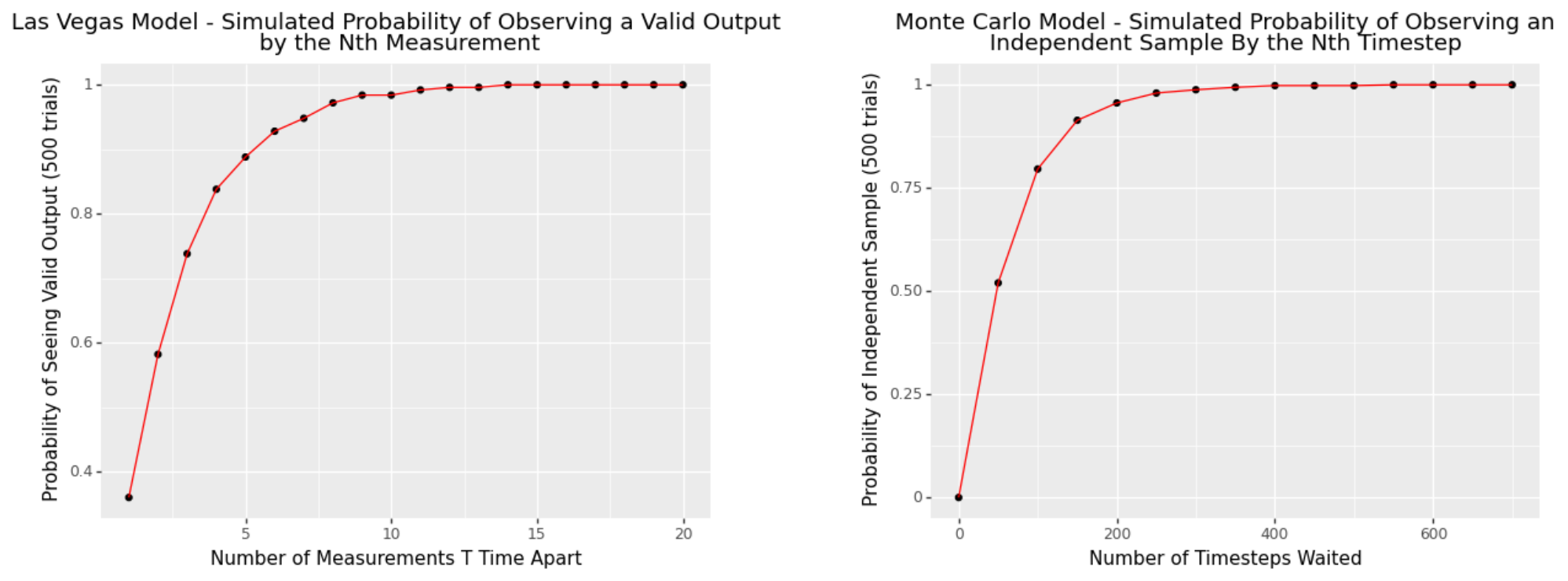}
    \caption{
    Empirical plots built by simulating our general construction on a reversible Turing machine. Las Vegas: the probability that an observer has found a new sample by each attempt.
    Monte Carlo: the probability that a sample is independent of the previous one over time.
    }
    \label{fig:plots}
\end{figure}
\section{Concrete Turing Machine Construction}
\label{sec:concrete}

In \cref{sec-LV} and \cref{sec-MC} we gave high level constructions of  adiabatic machines that require zero energy to run and give Las Vegas and Monte Carlo samples respectively from a desired distribution.
In this section, we make these constructions more explicit.
Specifically we will show how to modify Turing machine $M$ that maps a random seed in $\{0,1\}^r$ as input on its tape to a multi-tape Turing machine $M^*$ that can sample from the same distribution in either a Las Vegas or Monte Carlo fashion.
The general strategy directly follows from our results in \cref{sec-LV} and \cref{sec-MC}.
\begin{definition}
    An $m$ tape \emph{Turing machine} $M$ is a tuple $(S, \Sigma, f)$ where $S$ is a set of configurations and $\Sigma$ is a set of symbols for the tape.
    $f$ is a set of rewrite rules of the form $(s, \sigma) \to (s', \sigma')$ or $s \to (s', d)$, where $s, s' \in S$, $\sigma, \sigma' \in \Sigma^m$ and $d \in \{L,\emptyset,R\}^m$.
    
    In each timestep, the machine reads the values on the tape at each of its heads and its configuration and then applies some rewrite rule in $f$ that matches this configuration. Rewrite rules of the form $(s, \sigma) \to (s',\sigma')$ change the configuration and the values at the current locations on the tapes. Rules of the form $s \to (s',d)$ change the configuration and cause the $i$'th tape head to move one space to the left when $d_i=L$, one space to the right when $d_i=R$, or stay in the same place when $d_i= \emptyset$.
\end{definition}
\paragraph*{Our starting machine}
We will make a couple of simplifying assumptions about the machine $M$ that will make the proofs easier to follow; however, they are not strictly necessary for our argument to work.
\begin{enumerate}
    \item $M$ has an input tape and an output tape and starts with the input seed $b \in \{0,1\}^r$ written on its input tape with blank characters on either end of the input.
    \item $M$ is a reversible Turing machine. This can be achieved by applying Bennett's reversible pebble
    game construction~\cite{Ben89} to $M$.
    \item $M$ starts in a state $A$
    with the head at the far-right of the input tape (looking at the first blank). The output tape starts empty.
    \item $M$ produces the output in time at most $T$, which is known to this construction. We can get the same results for a slightly different construction if $T$ is not known ahead of time, but instead we have the guarantee that all paths in $M$ have the same length.\footnote{Moreover, a slight modification of the construction carries through even if we merely have the guarantee that $M$ halts on all inputs (i.e. if we don't know $T$), however in that case we loose the probability bounds in \cref{thm:LV-TM}. See \cref{sec:discussion} for further discussion.}
\end{enumerate}

\subsection{Las Vegas Sampling}
\begin{theorem}\label{thm:LV-TM}
Let $M = (S, \Sigma, f)$ be a (reversible) Turing machine that---when given a uniformly random bitstring over $\{0,1\}^r$ as input---produces a sample from a distribution $\calD$ in time at most $T$.
Then there is a Turing machine $M^*$ that is an adiabatic $(\Theta(T^2+r^2),3/4)$-efficient Las Vegas sampler of $\calD$.
\end{theorem}
\begin{proof}
We will convert machine $M$ to machine $M^*$ by defining a sequence of hybrid machines until the state space of our machine matches the general construction shown in \cref{fig:LV-general}.
Let $r$ be the length of the random input to $M$ and $T$ be an upper bound on the number of steps that $M$ takes to produce its output on any input.
Let $T_{\text{comp}}$ be the smallest power of two that is at least as large as $T$ and $T_{\text{out}}$ be the smallest power of two that is at least as large as $2T_{\text{comp}}+r+1$.
Note that $T_{\text{out}}$ and $T_{\text{comp}}$ are both $O(T+r)$.
We will start by extending $M$ to machine $M_1$ by adding the output holding region.
In other words, we will artificially manipulate the number of steps so that our machine takes the same number of steps for all input bitstrings and then further extend the length of the computation so that the machine is likely to observe the output when sampled.

We will do this by augmenting the machine with two new tapes designed to contain reversible counters.
We will refer to these tapes as the computation-counting and output-counting tapes.
The computation-counting tape starts with the value $0$ expressed as a $\log_2 T_{\text{comp}}$ bit binary number (surrounded by $\_$ symbols) with the least significant bit on the right and the head starts to the right of this number.
The output-counting tape starts with the value zero expressed as a $\log_2 T_{\text{out}}$ bit binary number with the head to the left of the number.
For each rewrite rule in $f$ that changes the state of $M$ to to some state $A$, we change the target state of that rule to a new state $\alpha_A$.
We additionally add new rewrite rules to $f$ for each configuration that was terminal in $M$ to map that state to a new state $\alpha_A$.
We add the following auxiliary states and rules that only act on the computation-counting tape:
\begin{align*}
\alpha_A &\to (\beta_A, L)&
\gamma_A &\to (\xi_A, R)\\
(\beta_A, 1) &\to (\alpha_A, 0)&
(\xi_A,0) &\to (\gamma_A, 0)\\
(\beta_A, 0) &\to (\gamma_A, 1)&
(\xi_A,\_) &\to (A,\_)\\
(\beta_A, \_) &\to (\omega, \_)
\end{align*}
Adding the above rules creates a (reversible) counter in the computation-counting tape that will increment after simulating each step of $M$.
When $M$ would have halted, the additional rules cause the machine to instead increment the counter.
These new rules give us a machine where a state is terminal iff the state is $\omega$ and the computation always runs for the same number of steps regardless of the input.
Now we can add more rules that transition out of $\omega$ and use the output-counting tape to produce our output holding region.
Specifically the new rules for this act only on the output-counting tape.
\begin{align*}
    (\omega, \_) &\to (\gamma, \_)&
    \gamma &\to (\xi,R)\\
    (\xi,0) &\to (\gamma,0)&
    (\xi,\_) &\to (\chi, \_)\\
    (\chi,\_) &\to (\alpha, \_)&
    (\beta,1) &\to (\alpha, 0)\\
    \alpha &\to (\beta, L)&
    (\beta,0) &\to (\gamma, 1)\\
    (\beta,\_) &\to \perp
\end{align*}
Where $\perp$ is a new halting configuration for the machine.
We add $\chi$ as an additional state this counter goes through so that the number of steps between each time this counter is incremented is the same as the number of steps for the computation-counting tape.
If instead of knowing the value of $T$ when designing this machine we have the guarantee that all chains in $M$ have the same length, we can instead design the counters so that we first count up to $T$ on an initially blank computation-counting tape as we perform our computation and then we count up to $2T+r$ on an initially blank output-counting tape.
This requires that all chains of $M$ have the same length to sample from the desired distribution, but removes the requirement of knowing $T$ ahead of time.

Together, these additional states and modified transitions give us a reversible machine $M_1$ where all computation paths take the same number of steps.
The state-space of the machine contains sufficient output-holding states for our general Las Vegas construction.
Let $\overline{M_1}$ be the \emph{chiral inversion} of $M_1$.
Specifically for each transition rule of $M_1$ that moves the heads of the Turing machine $\overline{M_1}$ moves the heads of all the tapes except for the computation-counting, output-counting, and output tapes in the opposite direction.
We note that $M_1$ behaves the same as $\overline{M_1}$ if the starting configurations of all other tapes is inverted around the heads.
We will have two copies each of $M_1$ and $\overline{M_1}$ denoted as $M_1^{0}, M_1^{1}, \overline{M_1^0},$ and $\overline{M_1^1}$.
Each state $A$ will be denoted $A^0, A^1, \overline{A^0}$ and $\overline{A^1}$ for these respective machines.

We will start constructing $M^*$ by combining the states and transition rules of $M^0_1, M^1_1, \overline{M^0_1}$ and $\overline{M^1_1}$.
Now to connect the states of the four machines, we will implement the randomizer from our general machine construction.
We add transition rules that will randomize the input before leading to the starting states of these machines (lets consider this state in $M_1$ to be $A$).
More specifically we add the following new states and transition rules that act only on the input tape:
\begin{align*}
    (\alpha,0) &\to (\beta, 0)&
    (\alpha,0) &\to (\beta, 1)\\
    (\alpha,1) &\to (\beta, 0)&
    (\alpha,1) &\to (\beta, 1)\\
    (\alpha,\_)&\to (A^0,\_)&
    (\alpha,\_)&\to (A^1,\_)\\
    (\overline{A^0}, \_) &\to (\beta, \_)&
    (\overline{A^1}, \_) &\to (\beta, \_)\\
    \beta &\to (\alpha, R)
\end{align*}
These states implement the randomizer in the general machine Las Vegas construction with the states of $M_1^0$ and $M_1^1$ acting on the (forward) input like the states right of the randomizer while the states of $\overline{M_1^0}$ and $\overline{M_1^1}$ act on the (backwards) input like the states left of the randomizer.
All that remains to emulate the genearal machine construction is to allow the states in $M_1^0$ / $M_1^1$ and $\overline{M_1^0}$ / $\overline{M_1^1}$ to transition to one another.
We can achieve this by taking each transition rule from each of these machines and making a copy of it that maps to the equivalent state in the other machine.
For example if $M_1$ had the transition rule $A \to (B,R^m)$ then we would add the four rules:
\begin{align*}
    A^0 &\to (B^1,R^m)&
    A^1 &\to (B^0,R^m)\\
    \overline{A^0} &\to (\overline{B^1},L^m)&
    \overline{A^1} &\to (\overline{B^1},L^m)
\end{align*}

This machine $M^*$ has the same state-space as the general machine for our Las Vegas construction.\footnote{The output holding regions may be longer than in the general machine construction, but this only improves the probability that each measurement gives a new sample.}
Specifically we can view $M^*$ as a general machine where the tapes, head locations, and state are in the work registers, the output register is the output tape.
The metadata is set to $-1$ (or $1$) when the machine is in states corresponding to incrementing the output-counting tape in the chirally inverted (non-inverted) machines and is otherwise set to~$0$.
By \cref{lem:gen-lv-sampler}, this makes $M^*$ an $(\Theta(T^2+r^2),3/4)$-efficient Las Vegas sampler.
\end{proof}

\subsection{Monte Carlo Sampling}
\begin{theorem}\label{thm:MC-TM}
Let $M = (S, \Sigma, f)$ be a (reversible) Turing machine that---when given a uniformly random bitstring over $\{0,1\}^r$ as input---produces a sample from a distribution $\calD$ in time at most $T$.
Then we can construct a Turing machine $M^*$ that is an adiabatic $(\Theta(T^2+r^2),1/2)$-efficient Monte Carlo sampler of $\calD$.
\end{theorem}

\begin{proof}
Let $M_1$ be the similarly named machine $M_1$ we generated in the proof of \cref{thm:LV-TM} except that we set $T_{\text{comp}}$ and $T_{\text{out}}$ to both be the smallest power of two larger than $T+r$ and initialize the computation-counting tape to be initialized with $r+1$ expressed in binary instead of zero.
This will make the length of the computation plus randomizer the same as the output holding region, which is essential for the functioning of our general construction.

We create two copies of $M'$ denoted as $M^0_1$ and $M^1_1$ where each state $A$ in $M'$ is denoted $A^0$ and $A^1$ for the respective machines.
We will start constructing $M_2$ by combining all states and transition rules in $M^0_1$ and $M^1_1$.
Now we connect the states of these two machines similarly to the Las Vegas construction by adding the same transition rules that will randomize the input. 
Let $A^0$ and $A^1$ be the starting configurations of their respective machines.
Then we add the following states and transitions:
\begin{align*}
    (\alpha, 0) &\to (\beta,0)&
    (\alpha, 0) &\to (\beta, 1)\\
    (\alpha, 1) &\to (\beta, 0)&
    (\alpha, 1) &\to (\beta, 1)\\
    (\alpha, \_) &\to (A^0, \_)&
    (\alpha, \_) &\to (A^1, \_)\\
    \beta &\to (\alpha, R)
\end{align*}
Similarly to the Las Vegas construction, we also add  transitions that allow the states from $M^0_1 / M^1_1$ to transition to one another during each step.

Now we have a Turing machine $M_2$ whose state-space graph looks like one of the sub-machines in our generic Monte Carlo construction.
All that remains is to combine two copies of machine $M_2$ so that one of them always has the output on its output tape.
Let $M^0_2 = (S^0, \Sigma, f^0)$ and $M^1_2 = (S^1, \Sigma, f^1)$ be two copies of the machine $M_2$.
We will build machine $M^* = (S^0\times S^1, \Sigma, f^*)$ that combines these machines.
This machine will have two copies of each tape in $M_2$ that are used for running $M^0_2$ and $M^1_2$.
For each pair of transition rules from $f^0$ and $f^1$, $f^*$ has a transition rule that implements the rule from $f^0$ forwards and the rule from $f^1$ backwards.\footnote{For example if these were one tape Turing machines, $f^0$ contained transition $(A^0, \alpha^0) \to (B^0, \beta^0)$, and $f^1$ contained the transition $(C^1, \gamma^1) \to (D^1, \delta^1)$ then $f^*$ would have the transition $((A^0, D^1), (\alpha^0, \delta^1)) \to ((B^0,C^1), (\beta^0,\gamma^1))$.}
$M^*$ would start in some initial configuration of the head and tapes for $M^0_2$ and some terminal configuration of the head and tapes for $M^1_2$.

The generated machine $M^*$ has a state space exactly like that of the general Monte Carlo sampler we described in \cref{sec-MC}.
We note that the meta-data register from the general machine construction can be simulated by reading if the state of $M^*$ corresponds to either sub-machine $M^0_2$ or $M^1_2$ being in a state where it is incrementing the output-counting tape, which can be expressed as a single bit.
The output can then be read by looking at the output tape for the appropriate sub-machine.
Thus by \cref{lem:gen-mc-sampler}, this makes $M^*$ an $(\Theta(T^2 + r^2), 1/2)$-efficient Monte Carlo sampler.
\end{proof}
\section{Discussions and Open Questions}\label{sec:discussion}
\subparagraph*{On the need for reversibility}
Our constructions start with the assumption that the machine $M$ is a logically reversible Turing machine.
Thanks to constructions like \cite{bennett1973logical}, we can safely make this assumption with only small overheads to the time and space complexities of $M$.
It is natural to ask whether instead we can start with irreversible computation, and adjust $\Delta G$ of states (see \Cref{append:thermodynamics}) so that we obtain an unbiased one-dimensional random walk over the layers. 
Unfortunately, as \cref{fig:irreversible-states} shows, a random walk over the layers will hit obstacles called \emph{traps}: states that are not reachable from any valid input by applying the irreversible transition function $f$.
As in general there can be exponentially many such traps in the length of the computation, it could take an exponential time for the machine to return to the randomizer.
This, in turn,
could result in an exponential increase in the time the observer needs to wait before measuring the device.

\begin{figure}
    \centering
    \includegraphics[width=.6\textwidth]{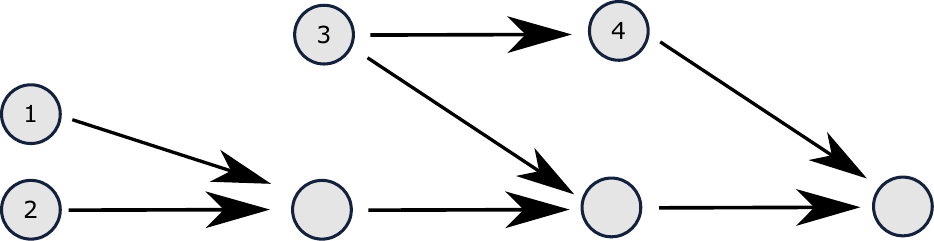}
    \caption{An irreversible computation. States 1, 2 are inputs while 3, 4 are \emph{traps} that do not come from an input.}
    \label{fig:irreversible-states}
\end{figure}

\subparagraph*{Unknown bounds on running time}
Our Las Vegas and Monte Carlo constructions make the assumption that for the input machine $M$, the quantity $T$ representing the maximum number of steps for $M$ to halt on any input of length $r$ is known to our construction.
While we are able to prove upper bounds on the running time for specific algorithms, as pointed out in \cite{norton2013brownian}, this is in general impossible for an arbitrary Turing machine $M$.
Fortunately there are two ways to modify our constructions for when $T$ is unknown.

First if we know that $M$ takes the same number of steps to run on all length $r$ inputs, then we can modify $M$ by adding a new tape that counts the number of steps during the computation.
The length of the output holding region can then be programmatically determined by performing arithmetic on this tape and then having it count down to make the output holding states.
If $M$ takes a different number of steps on different inputs then the probabilities of being in each output holding region may be different.
Thus this trick will give samples from a different distribution than desired.

Second if we are willing to wait $\Omega(4^r T^2)$ time between samples, we can artificially make all inputs take the same number of steps.
We do this by modifying $M$ to compute what it would output on all inputs and then output the appropriate value for its actual input.
Unfortunately this construction still has a problem.
Without knowing the value of $T$ it is unclear how long the observer should wait between samples.


\subparagraph*{Open questions} 
Our constructions make the simplifying assumption that a machine can be designed so that all transitions are adiabatic.
Any physical implementation of our constructions would likely involve some physical defects that would bias some transitions over others.
It would be interesting to analyse the behavior of our constructions under adversarial or random manipulation of the physical embedding.

Our constructions also require time proportionate to the square of the \emph{maximum} time an algorithm would take to produce a sample.
We think it would be interesting to either construct a machine where the time between samples only scales quadratically with the \emph{average} time to produce a sample or to prove that in general, such a machine cannot exist.
Finally it may be possible to prove that this time (whether maximum or average) cannot have sub-quadratic scaling for any reusable adiabatic sampler.

\bibliographystyle{plainurl}
\bibliography{sources}
 \appendix

\section{Thermodynamics and Kinetics}\label{append:thermodynamics}

\subparagraph*{Thermodynamic Model (Boltzmann Distribution).}
We can think of the state space of any physical system as a graph where each node represents a physical configuration (macro-state) of that system and edges indicate which configurations are reachable in one step. 
These configurations may contain different amounts of free energy. 
For a configuration $A$, we denote its Gibbs free energy as $\Delta G_A$.
Following the canonical Boltzmann distribution (also known as the Gibbs distribution) \cite{landau2013statistical}, we say that at equilibrium (when time goes to infinity) the probability of being in any configuration $A$ is a function of the Gibbs free energy of $A$:
$$ Pr_{eq}[A] = \frac{e^{-\Delta G_A / kT}}{\sum\limits_{B\in V} e^{-\Delta G_B / kT}}$$
Normalizing over the partition function ($\sum_{B\in V} e^{-\Delta G_B / kT}$) makes the equilibrium probability of a particular configuration $A$ proportional to how favorable $A$ is with respect to all other configurations.
The following is immediate:

\begin{proposition}\label{prop:uniform-dist}
If all configurations of a physical system have the same free energy, then all configurations have the same equilibrium probability and the Boltzmann distribution becomes a uniform distribution. 
\end{proposition}

\subparagraph*{Kinetic Model (Metropolis Dynamics).}
Since we are interested in quantifying the efficiency of our constructions in terms of the time between samples, we need a kinetic model that is consistent with the above thermodynamic model. 
We consider the (continuous) time evolution of a general machine as a continuous-time Markov process over the graph that we call a physical embedding:


\begin{restatable}{definition}{embedding}
    Let $M = (\Sigma, f)$ be a general machine.
    A \emph{physical embedding} of $M$ is a directed graph $P_M = (V, E)$ with vertices $V = \Sigma$ and edges $E = \{(U,W) \mid W \in f(U) \text{ or } U \in f(W) \}$ (i.e., if $W$ is a successor configuration of $U$ then they are connected by forward and backward edges).  
    Each vertex $U$ is assigned a value $\Delta G_U \in \mathbb{R}$, and each edge $(U,W)$ is assigned the weight $\lambda_{U \to W}$ dictated by Metropolis dynamics~(see \Cref{eqn:metropolis} below).
\end{restatable}

Note that instead of Metropolis dynamics (defined below), our definition could have used any other kinetic rate law that is consistent with the Boltzmann distribution.\footnote{Our results will involve the setting where each configuration has the same $\Delta G$ and the system evolved via an unbiased random walk over configurations. 
As a result, our results do not depend on the particulars of the Metropolis dynamics.
}

Metropolis dynamics~\cite{metropolis1953equation} is a commonly used~\cite{schaeffer2015stochastic} kinetic law under which all favorable transitions occur at the same fixed rate, 
while the rate of unfavorable transitions scales according to the change in free energy. 
For two configurations $A$ and $B$, let $\lambda_{A\rightarrow B}$ be the rate parameter dictating the rate at which $A$ transitions to $B$.
Then metropolis dynamics corresponds to:
\begin{equation}
\label{eqn:metropolis}
\lambda_{A\rightarrow B} =
\begin{cases}
k \frac{e^{-\Delta G_B / kT}}{e^{-\Delta G_A / kT}} & \text{if } \Delta G_A < \Delta G_B\text{ (i.e., $A \rightarrow B$ is unfavorable)}\\
k & \text{otherwise}
\end{cases}
\end{equation}
It is known that Metropolis dynamics implies Boltzmann distribution thermodynamics, while satisfying the physically realistic assumption that all rates are bounded from above independent of the free energies~\cite{metropolis1953equation}.
\begin{proposition}\label{prop:random-walk}
    If all configurations of a physical system have the same free energy, then every transition occurs at the same rate and Metropolis dynamics becomes an unbiased random walk over the configurations.
\end{proposition}

When drawing examples of physical embeddings (e.g., \Cref{fig:LV-simple} and \Cref{fig:LV-general}), 
we draw undirected graphs 
with the understanding that each edge corresponds to two directed edges.

\subparagraph*{Adiabatic (Zero Energy) Computation}
In this paper, we consider machines that are adiabatic in the sense that they release or consume no free energy as heat during execution.

\begin{definition}\label{def:zero-energy}
    Let $P_M = (V,E)$ be a physical embedding of general machine $M$.
    We call $P_M$ an \emph{adiabatic machine} if $\Delta G_U = \Delta G_W$ for all $U, W \in V$.
\end{definition}

Such a machine has kinetic and thermodynamic behavior that is easily reasoned about.
From Definition~\ref{def:zero-energy}, Propositions~\ref{prop:uniform-dist} and~\ref{prop:random-walk} (and the fact that Metropolis dynamics imply Boltzmann distribution at equilibrium) we get the following corollary:

\adiabaticevolution

\end{document}